\documentclass[runningheads]{llncs}

\usepackage[T1]{fontenc}
\usepackage[utf8]{inputenc}
\usepackage{booktabs}
\usepackage{amsfonts}
\usepackage{amsmath}
\usepackage{graphicx}
\usepackage{subcaption}
\usepackage{multirow}
\usepackage{xcolor}
\usepackage{hyperref}
\usepackage{url}
\usepackage[letterpaper,margin=1in]{geometry}
\let\oldthebibliography\thebibliography
\renewcommand{\thebibliography}[1]{%
  \oldthebibliography{#1}%
  \setlength{\itemsep}{0pt}\setlength{\parsep}{0pt}}
\usepackage{enumitem}
\setlist{itemsep=2pt,parsep=0pt,topsep=4pt,partopsep=0pt}
\graphicspath{{figures/}}

\begin{document}

\title{Discovering Dual-Origin Slow Wind from Solar Orbiter with
Self-Supervised Contrastive Learning}

\titlerunning{Discovering Dual-Origin Slow Wind from Solar Orbiter}

\author{Henry Han\thanks{Corresponding author.} \and Jorge Yero Salazar}
\authorrunning{H. Han and J. Yero Salazar}

\institute{Department of Computer Science, Baylor University,\\
One Bear Place, Waco, TX 76798, USA\\
\email{henry\_han@baylor.edu}}

\maketitle

\begin{abstract}
Whether the slow solar wind originates from one coronal source or two distinct
channels remains a central open question in heliophysics. Because these
populations arrive at nearly the same bulk speed and differ mainly in
heavy-ion composition, they must be separated without labels. We present
Solar-CDC, a self-supervised contrastive deep clustering (CDC) framework that
maps plasma observables to a latent space via a Transformer encoder, optimizes
a triplet margin loss, and iteratively updates pseudo-labels via $k$-means.
Theoretically, we prove a hard limit on neighborhood-preserving embeddings
such as t-SNE and UMAP. If mixing is measured by the fraction of a point's
nearest neighbors belonging to the other population, preserving the neighbor
graph leaves this fraction unchanged; preserving all but a fraction
$\varepsilon$ of the links shifts it by at most $\varepsilon$. Neither bound
depends on the embedding dimension: such methods cannot create separation the
measurements lack. A margin objective, however, rewrites the graph and drives
the fraction to zero. Empirically, on $30{,}602$ Solar Orbiter observations,
thirty combinations of dimensionality reduction and clustering peak at a
silhouette of $0.454$, whereas Solar-CDC reaches $0.869$. Escaping this limit
is not enough: TriMap also optimizes triplets and reaches $0.824$, yet its
clusters score below chance against the published composition taxonomy.
Solar-CDC instead recovers clusters with mean charge-state ratios of $0.080$,
$0.160$, and $0.400$, placing the intermediate population squarely inside the
window associated with coronal-hole boundaries. Even when the defining
charge-state ratio is withheld from the inputs entirely, the model still
recovers the taxonomy defined on it. Importantly, a learning loss recovers
physical populations only when driven by dynamically updated physically-aware
clusters rather than distances.

\end{abstract}

% =============================================================================
\section{Introduction}
The solar wind, a constant, supersonic stream of magnetized plasma from the Sun's corona, is the primary way the Sun interacts with the heliosphere. When this plasma reaches Earth, it triggers geomagnetic storms, substorms, and auroral activity. These events can damage satellites, disrupt power grids, and degrade radio communications  \cite{solar_and_magnetic_field,solarwind_geomagnetic,solar_space_weather}. A central unresolved question in heliophysics is whether the ``slow'' solar wind originates from a single coronal source or two distinct channels.

\textit{The Dual-Origin Slow Wind Problem.} The origin of the fast solar wind
is settled; the source of the slow wind is not. What separates the candidates
is composition rather than speed. Fast wind ($v_p > 600$~km/s) originates from
coronal holes, regions of open magnetic field, and carries low proton density
and a low oxygen charge-state ratio ($\mathrm{O}^{7+}/\mathrm{O}^{6+} < 0.145$)
that reflects the cool, rapidly diverging environment in which it forms
\cite{solar_fast_and_slow,lepri2013solar}. Slow wind ($v_p < 500$~km/s) is
denser and more variable, and its candidate sources include helmet-streamer
tips, coronal-hole boundaries and active-region peripheries. A recently
identified \textit{Alfv\'enic slow wind} travels at slow-wind speed while
carrying the composition of the fast wind, which points to an origin near
coronal-hole boundaries \cite{damicis2021alfvenic}. Alfv\'enicity here is the
degree to which velocity and magnetic-field fluctuations correlate, the
signature of outward-propagating Alfv\'en waves.

\begin{figure}[!htbp]
\centering
\includegraphics[width=\linewidth]{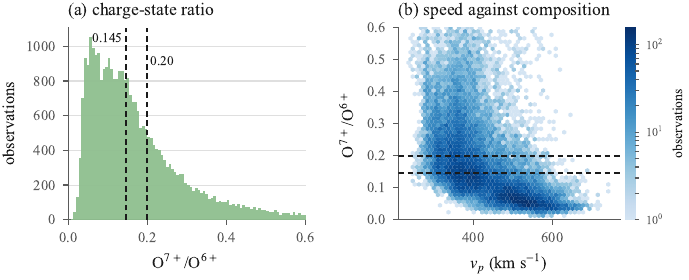}
\caption{Why a threshold on one variable is not enough. (a) The
$\mathrm{O}^{7+}/\mathrm{O}^{6+}$ distribution over the 30{,}602 Solar Orbiter
observations, with the coronal-hole ($0.145$) and streamer-belt ($0.20$)
reference values marked. (b) Bulk speed against composition: the two slow-wind
channels are stacked in composition at essentially the same speed, so a
velocity cut cannot separate them.}
\label{fig:problem}
\end{figure}

Turning that picture into a task on the data gives three parts. The first is to separate fast coronal-hole wind from all slow wind. The second is to distinguish streamer-belt slow wind, of closed-field origin, from boundary-type slow wind, of open-field origin. The third is to do both from in-situ plasma measurements alone. Solving (ii) would
support the \emph{dual-origin hypothesis} and constrain coronal magnetic field
models. Part (ii) is the hard one. The two populations overlap in bulk speed, the flow speed of the plasma as a whole rather than the thermal motion of its particles. The discriminating signal therefore lies in the joint heavy-ion composition, a multivariate fingerprint that no single-threshold scheme can exploit
(Fig.~\ref{fig:problem}).

\label{sec:intro-physics}%
\textit{Physical Basis for a Three-Population Model.} Three independent physical
mechanisms justify a three-population model ($k = 3$). First, coronal magnetic
topology partitions the solar atmosphere into open-field coronal holes and
closed-field streamer belts, establishing a baseline of $k \ge 2$. Second, the internal structure of coronal holes introduces a vital subdivision. Rapidly expanding interior flux tubes produce fast, cool wind, whereas slowly expanding boundary flux tubes produce the Alfv\'enic slow wind \cite{damicis2021alfvenic}. Its composition diverges from closed-field streamer-belt plasma, which raises the bound to $k = 3$. Third, charge-state freeze-in ratios
($\mathrm{O}^{7+}/\mathrm{O}^{6+}$) fixed at $1$--$2\,R_\odot$ encode the source electron temperature $T_e$. They map these topological regions to distinct, non-overlapping ranges: $<0.145$ (coronal-hole interior, $T_e < 1.4$~MK)
\cite{xu2015new}, $0.145$--$0.20$ (coronal-hole boundary,
$T_e \approx 1.5$--$1.6$~MK) \cite{damicis2021alfvenic}, and $>0.20$ (streamer
belt, $T_e > 1.6$~MK) \cite{lepri2013solar}.

\textit{Limitations of Existing Methods.} Traditional classification applies expert thresholds to one or two scalar
parameters \cite{xu2015new}. Such schemes are reproducible but cannot
resolve the two slow-wind populations, whose defining differences lie in the
joint $(\mathrm{O}^{7+}/\mathrm{O}^{6+},\,\mathrm{C}^{6+}/\mathrm{C}^{4+},\,
\mathrm{Fe/O})$ space rather than in any single variable
\cite{xu2015new,roberts2020determination}.
Unsupervised methods such as
$k$-means \cite{solar_kmeans} and dimensionality-reduction pipelines
\cite{carpenter2023dimension} have been applied to solar wind data, but the
projection is computed first and the clustering runs on whatever it produces.
PCA keeps the directions of largest variance; t-SNE and UMAP keep each point
near the neighbors it had in the input. Neither asks which observations ought to
end up together, and with no labels there is nothing to correct the choice
against afterwards. This matters most for unlabeled data such as ours, where no
ground truth exists to reveal that a projection has hidden the boundary. Sec.~\ref{sec:theory} shows the consequence: a projection
of that kind cannot sharpen a boundary the input metric does not already
contain, whatever its target dimension.

\textit{Our Solution.} We propose Solar-CDC, a self-supervised contrastive deep clustering (CDC) method. Using only unlabeled in-situ measurements, it learns a representation for plasma observations. In this new space, the distance between two observations reflects the similarity of their coronal source rather than their bulk speed.
It separates the populations more
sharply than any baseline we tested (Sec.~\ref{sec:results}). On Solar Orbiter it splits the slow wind in two. The halves differ by $4$~km\,s$^{-1}$ in speed but by a factor of $2.5$ in oxygen charge state, and one of them falls inside the window associated with coronal-hole boundaries. That is the split the dual-origin hypothesis predicts, found without labels.
  
Solar-CDC does not rely on a fixed dimension-reduction projection. Instead, it clusters observations by alternating between two interdependent tasks: updating the data representation based on the current grouping, and re-forming the groups within that new representation. A Transformer encoder maps each observation into a latent space to apply the contrastive objective. During each training round, the model pulls each observation closer to its current group members. Simultaneously, it pushes the observation away from members of other groups. This process adjusts the encoder's weights, causing the latent representation to evolve without relying on external true-source labels. As the latent layout changes, the grouping is recomputed. The contrastive pushing and pulling then restarts using these updated groups. This self-supervised cycle repeats until the layout stabilizes and re-clustering returns identical groups. 

However, this iterative process lacks an external anchor and is highly sensitive to initial conditions. Starting from random groups would lock the model into an arbitrary pattern. To prevent this, we derive the initial groups directly from the physical plasma measurements (called 'warm-up'). This initialization ensures that the final clusters are far more likely to correspond to true solar wind populations.

\textit{Why This Addresses the Dual-Origin Question.} The two slow-wind channels overlap in bulk speed and density, separating only in composition (Fig.~\ref{fig:problem}). While simple velocity thresholds cannot detect this boundary, our method learns directly from raw observables to correctly partition the composition space. Initializing the network with these observables anchors the clusters to the physical charge-state structure rather than random geometric noise, and the alternating loop sharpens the boundaries entirely without labels. Sec.~\ref{sec:results} validates both the necessity of this physical initialization and the accuracy of the resulting coronal source taxonomy.

\textit{Contributions.} This study makes the following major contributions:

\begin{enumerate}
\item \textit{Theoretical separability bound (Sec.~\ref{sec:theory}):} We derive a bound determining which data representations can separate these populations based on the objective rather than the target dimension. Embeddings preserving the $\kappa$-nearest-neighbor graph leave the cut fraction unchanged (or bounded by $\varepsilon$), while triplet margin embeddings drive it to zero. This holds across all dimensions.

\item \textit{Empirical validation:} We validate the Solar-CDC bound across 39 baselines. As predicted, neighborhood-preserving methods (e.g., t-SNE, UMAP) fail to separate populations, yielding cut fractions near the observable space baseline ($0.064$--$0.084$). Triplet-based TriMap escapes this band (silhouette $0.824$), and Solar-CDC drives the cut fraction to near-zero ($0.001$).

\item \textit{Geometry versus physics:} We demonstrate that geometric separation from dimension reduction is necessary but insufficient for physical relevance. While TriMap clusters tightly, its partition is physically meaningless against the established taxonomy (adjusted Rand index $-0.021$). Solar-CDC uniquely achieves both tight separation (silhouette $0.869$) and true physical alignment (NMI $0.422$).

\item \textit{Recovery of solar wind populations:} Solar-CDC identifies clusters with mean charge-state ratios of $0.080$, $0.160$ (aligning with coronal-hole boundaries), and $0.400$. It separates two slow-wind populations differing by merely $4$~km\,s$^{-1}$ in speed but a factor of $2.5$ in charge state. Even with $\mathrm{O}^{7+}/\mathrm{O}^{6+}$ entirely withheld, Solar-CDC recovers the taxonomy (NMI $0.321$, against $0.000$ for a random labeling).

\item \textit{Architectural insights:} Ablation reveals the latent dimension strictly governs performance ($p = 1.8\times10^{-4}$), while a single attention head suffices ($p = 0.19$). Furthermore, training solely on the seven physical observables yields significantly better clustering ($p = 1.2\times10^{-3}$) than appending metadata like observation times or archived labels.
\end{enumerate}

\section{Related Work: Why Standard Tools Fall Short}

The two candidate slow-wind populations travel at the same speed and arrive
mixed together. This section reviews how the problem has been approached and
why each family of standard tools stops short of separating them.

\textit{The Limits of Traditional Classification.} Solar wind has historically
been separated by applying thresholds to bulk properties such as proton
specific entropy or Alfv\'en speed \cite{xu2015new}. When unsupervised
clustering was introduced, $k$-means on bulk parameters separated fast wind
from slow wind \cite{solar_kmeans}, and adding magnetic variance recovered
classes consistent with Alfv\'enic and non-Alfv\'enic slow wind
\cite{roberts2020determination}. The difficulty is that the two candidate
slow-wind populations overlap almost perfectly in bulk speed and density.
Separating them by a velocity threshold is like sorting red apples from green
ones with a scale: the instrument is accurate, but the quantity it measures
does not carry the distinction. Our prior work explored stacked dimensionality
reduction \cite{carpenter2023dimension} and composition-aware classification
for Solar Orbiter \cite{zhao2024classification}. The latter supplies the measurements used here, and it is the dependence of such schemes on predefined class labels that we set out to remove.

\textit{The Dimensionality Reduction Trap.} If the composition data is
high-dimensional, an obvious move is to project it into two dimensions with a
manifold learning tool such as t-SNE \cite{tsne}, UMAP \cite{umap}, or PHATE
\cite{phate}. Each is designed to preserve the local neighborhood structure of
the input space. That design is the trap: if two slow-wind populations are
tangled as mutual neighbors in the raw observable space, which they are, a
neighborhood-preserving projection keeps them tangled in the output. Such a
method does not push distinct physical populations apart; it takes a
lower-dimensional picture of the arrangement it was given.
Sec.~\ref{sec:theory} makes this precise and bounds how far the tangle can
be reduced. TriMap \cite{trimap} escapes the bound by optimizing triplet
relationships instead of neighborhoods, and Sec.~\ref{sec:external} shows what
that buys: a geometrically tight partition that does not correspond to the
composition taxonomy. These methods also produce a 2-D view, and
Sec.~\ref{sec:protocol} explains why cluster-quality scores computed inside
such a view cannot be compared with scores computed in a learned latent space.

\textit{The Missing Piece in Contrastive Learning.} Unlabeled data is
addressed in modern machine learning by self-supervised representation learning
\cite{ssl_survey,bachman2019representation} and deep clustering
\cite{deepclustering}. The usual engine is a triplet loss \cite{triplet_loss}
that pulls similar items together and pushes dissimilar ones apart, combined
with a Transformer encoder \cite{vaswani2017attention} to capture interactions
among features. Much of this literature was developed for computer vision,
where a positive pair is produced by augmenting an image: a crop, a flip or a
blur yields a second view of the same object. That device does not transfer
here. There is no crop or blur of a seven-dimensional composition vector that
leaves the physics intact, since altering a carbon charge-state ratio alters
what the measurement means. The augmentation-based branch of self-supervised
learning is therefore unavailable, and a positive must be a second real
observation rather than a modified copy of the first.

\textit{Enter Solar-CDC.} Separating the two slow-wind populations therefore calls for three things at once. The method must actively push overlapping populations apart, unlike a neighborhood-preserving projection. It must use no human-assigned labels, unlike a threshold scheme. And it must need no augmentation, unlike standard contrastive learning. Solar-CDC combines the three. It initializes its groups from the
physical observables and feeds them into an alternating triplet-loss loop, so
the supervisory signal is produced by the data itself.

\section{Solar Orbiter Data and Feature Space}
\label{sec:data}
We use in-situ measurements from the Solar Orbiter spacecraft \cite{zhao2024classification} spanning January 2022 through April 2023. The Solar Wind Analyser Proton-Alpha Sensor (SWA-PAS) provides bulk plasma properties, and the Heavy Ion Sensor (SWA-HIS) provides ion composition. The model input consists of seven physical observables:

\begin{enumerate}
\item \textit{Proton bulk speed} $v_p$ (km/s): Measures how fast the plasma moves past the spacecraft. It effectively separates fast wind from slow wind, but cannot distinguish the two slow-wind sources.
\item \textit{Proton number density} $N_p$ (cm$^{-3}$): Measures the number of protons per unit volume. The slow wind is consistently denser than the fast wind.
\item \textit{Oxygen charge-state ratio} $\mathrm{O}^{7+}/\mathrm{O}^{6+}$: The proportion of oxygen ions that have lost seven electrons instead of six. This ratio ``freezes in'' as plasma leaves the corona and does not change afterward, preserving a permanent record of the source region's temperature. Values below $0.145$ indicate a cool, fast coronal-hole origin \cite{lepri2013solar}.
\item \textit{Carbon charge-state ratio} $\mathrm{C}^{6+}/\mathrm{C}^{4+}$: An independent temperature reading. Because carbon freezes in over a different temperature range than oxygen, it provides complementary evidence of the source environment.
\item \textit{Carbon charge-state ratio} $\mathrm{C}^{6+}/\mathrm{C}^{5+}$: A third temperature reading sensitive to a narrower band, helping to resolve coronal sources that the other ratios place at nearly the same temperature.
\item \textit{Iron-to-oxygen ratio} Fe/O: Indicates the relative abundance of iron. Because iron ionizes more easily low in the atmosphere, an elevated ratio marks plasma that lingered on closed magnetic field lines before being released into the solar wind.
\item \textit{Mean oxygen charge state} $\langle Q_O \rangle$: Summarizes the entire oxygen distribution rather than just a single pair of ionization stages, providing the steadiest composition metric.
\end{enumerate}

\section{Methods}
This section introduces Solar-CDC and how it solves the dual-origin slow-wind challenge.

\subsection{Problem Formulation}
The boundary between the two slow-wind populations is invisible in bulk speed. The true coronal source of a plasma parcel cannot be measured by a spacecraft at all, so any operational boundary must be learned without ground-truth labels. Our goal is therefore a transformation that pulls apart
populations that are physically distinct but kinematically overlapping.

Formally, let $\mathcal{X} = \{x_n\}_{n=1}^{N}$ be the unlabeled set of in-situ
solar wind  observations, where $x_n \in \mathbb{R}^{p}$ and $p$ is the number of physical
observables (features) recorded per observation ($p = 7$ for the Solar Orbiter data of
Sec.~\ref{sec:data}). Rather than cluster in $\mathbb{R}^p$ directly, we learn a
parameterized encoder $f_\theta: \mathbb{R}^p \to \mathbb{R}^m$ into a latent
space where the Euclidean distance between two points reflects the similarity of
their coronal origin rather than their similarity in bulk speed. The latent
dimension $m$ is a free parameter and is not required to be smaller than $p$;
Sec.~\ref{sec:theory} shows that what decides separability is the objective, not
the dimension.

Simultaneously we seek a partition $\mathcal{C} = \{C_1, \ldots, C_k\}$ of the
latent codes into $k$ clusters, such that each $C_i$ recovers a physically
distinct plasma population: fast wind, boundary-type slow wind, and
streamer-belt slow wind. While the application here is the solar wind, nothing
in the method or in the bounds of Sec.~\ref{sec:theory} depends on the value of
$p$ or on the inputs being plasma observables.

\begin{figure}[!htbp]
\centering
\includegraphics[width=0.84\linewidth]{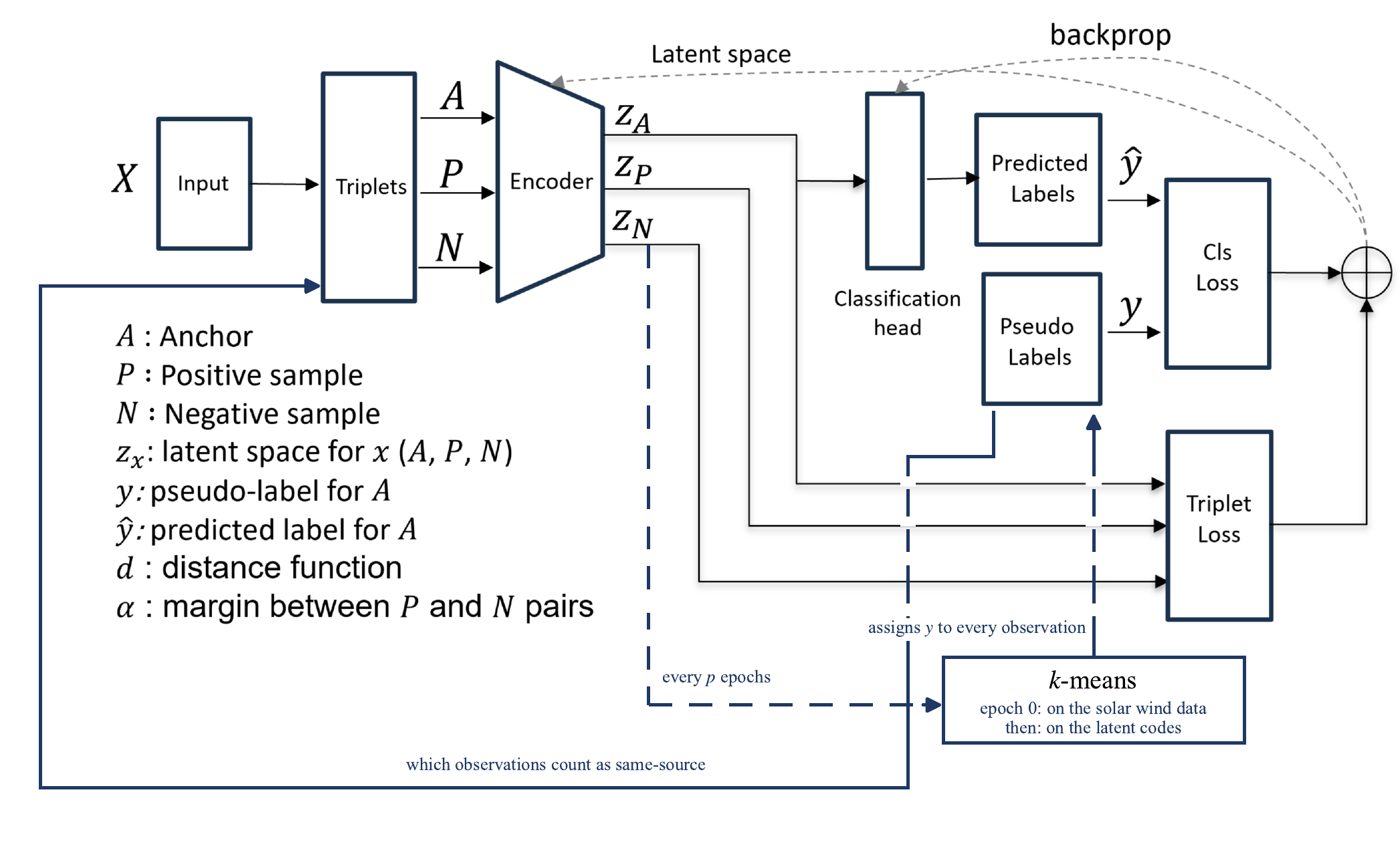}
\caption{Solar-CDC. Triplets drawn from the current pseudo-labels pass through
the shared encoder; the classification head is trained on those pseudo-labels
while the triplet term acts on the latent codes, and both gradients reach the
encoder. The lower half of the diagram is the part that makes the method
self-supervised: the pseudo-labels are not given, they are produced by
$k$-means, on the observables at epoch~0 and on the latent codes every $p$
epochs thereafter, and they decide in turn which observations are drawn as
positives and negatives.}
\label{fig:arch}
\end{figure}

\subsection{Solar-CDC}
\textit{Solar-CDC $=$ self-supervised contrastive learning $+$ deep clustering.}
Solar-CDC combines self-supervised contrastive learning \cite{ssl_survey} with
deep clustering \cite{deepclustering} to produce a physically aware partition of
the solar wind, directly addressing the dual-origin question.

It begins with a warmup phase: $k$-means applied directly to the standardized
observables supplies the initial pseudo-labels (Sec.~\ref{sec:warmup}). A
Transformer encoder then maps the observations into a latent space where the
contrastive objective is applied. For each observation (the anchor), the model
randomly draws a positive sample from the same pseudo-label group and a negative
sample from a different group. A triplet loss \cite{triplet_loss} then pulls the
anchor toward the positive and pushes it away from the negative by a fixed
margin. That margin is the engine of the separation. Rather than pulling points toward a cluster center, it imposes a relative distance constraint. Satisfying that constraint forces the neighbor links crossing the boundary to be broken,
the very links a neighborhood-preserving projection is obliged to keep
(Sec.~\ref{sec:theory}). A positive is therefore not an augmented view of the
anchor, as it would be in vision, but another real observation that the current
grouping treats as same-source.

Crucially, every $p$ epochs the pseudo-labels are refreshed by re-running
$k$-means on the updated latent codes. This creates a self-correcting feedback
loop: the grouping that defines the triplets is continuously refined by the very
representation it shapes. The cycle leaves a final latent space in which
observations sharing a coronal source cluster together.

This is how Solar-CDC addresses the core difficulty of the dual-origin problem.
By letting the contrastive loop work in the joint composition space, it separates two slow-wind streams that bulk speed cannot tell apart, placing the boundary where the charge-state structure puts it. It does so without manual velocity thresholds and without ground-truth labels.

\textit{Architecture and Triplet Loss.} As shown in Fig.~\ref{fig:arch}, each of the seven observables is treated as a sequence token. The Transformer encoder outputs a latent code $z_n \in \mathbb{R}^m$, and a linear classification head $g_W$ predicts its current pseudo-label $y_n$. For an anchor $x_n^A$, a positive $x_n^P$, and a negative $x_n^N$, the overall objective combines a cross-entropy term with the triplet margin:

\begin{equation}
\mathcal{L}(\theta, W) = \frac{1}{N}\sum_{n=1}^{N}
\Big[\underbrace{\ell_c\big(g_W(z_n^A),\, y_n\big)}_{\text{pseudo-label fit}}
+ \underbrace{\max\!\big(0,\, d(z_n^A, z_n^P) - d(z_n^A, z_n^N) + \alpha\big)}
_{\text{triplet margin}}\Big].
\label{eq:loss}
\end{equation}

Because the $k$-means targets are provisional (Fig.~\ref{fig:steps}), a fraction of them is inherently wrong at any point. The optimization thus operates in the regime of learning from mislabeled data in high dimensions \cite{han2024mislabeled}. The triplet term survives this because it acts as a hinge: once a triple satisfies the margin $\alpha$, it contributes no gradient. This limits how far incorrect pseudo-labels can distort the latent representation before the next $k$-means reassignment corrects them.

Consequently, two partitions coexist during training: the $k$-means assignments (which supply the provisional targets) and the head's $\arg\max$ predictions. The head's output serves as the final trained partition upon which all clustering metrics are computed. 

\textit{Implementation Details.} We train for 200 epochs using a batch size of 20{,}000, the AdamW optimizer at $10^{-3}$, and a margin $\alpha = 1$. The encoder uses one Transformer layer of width 32 with dropout $0.1$. The pseudo-label refinement occurs every $p = 4$ epochs, initialized from a random member of each previous cluster.

\begin{figure}[!htbp]
\centering
\begin{subfigure}[t]{0.49\linewidth}
\centering
\includegraphics[width=\linewidth]{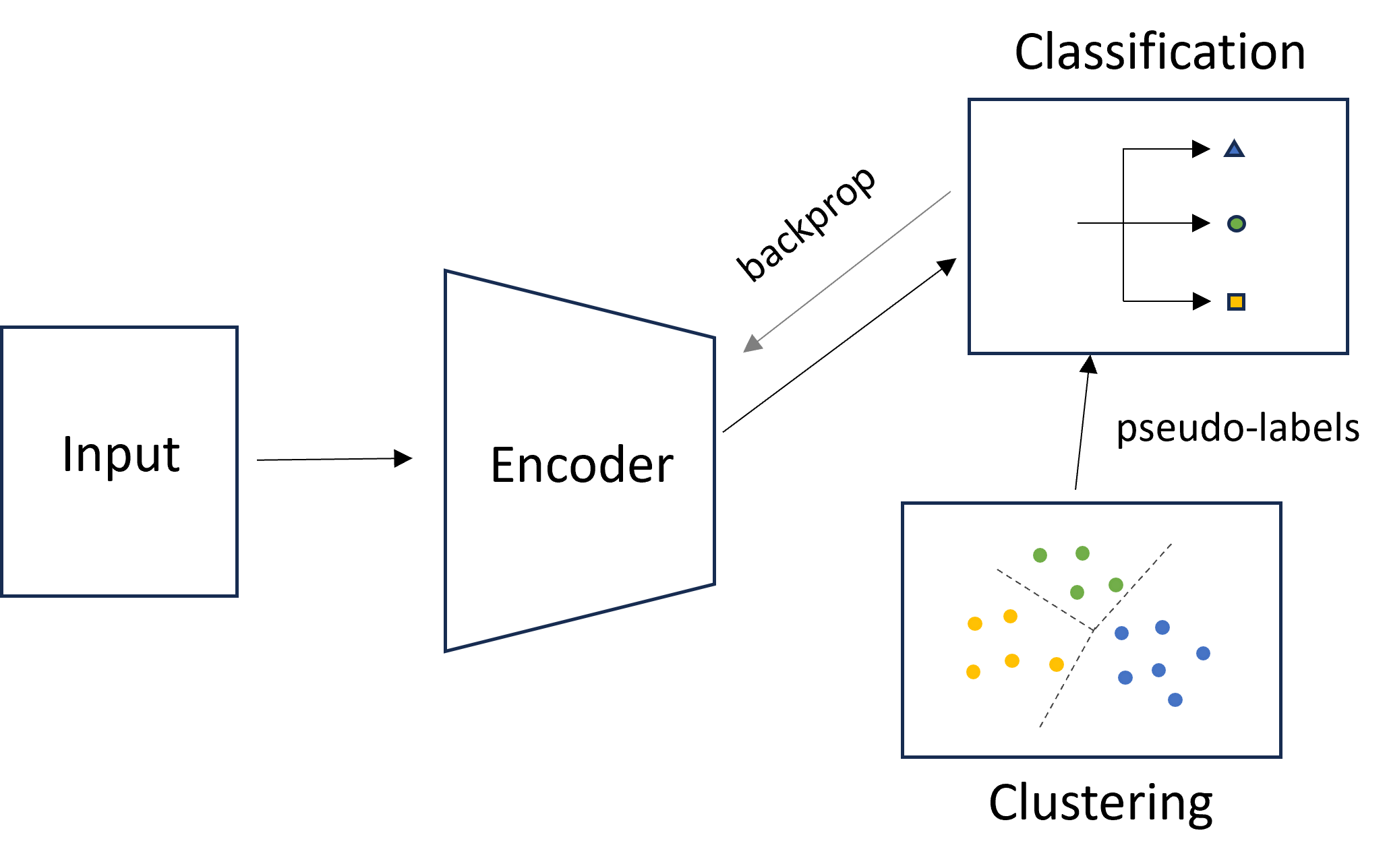}
\caption{Regular step: train on the current pseudo-labels.}
\label{fig:regular_step}
\end{subfigure}\hfill
\begin{subfigure}[t]{0.49\linewidth}
\centering
\includegraphics[width=\linewidth]{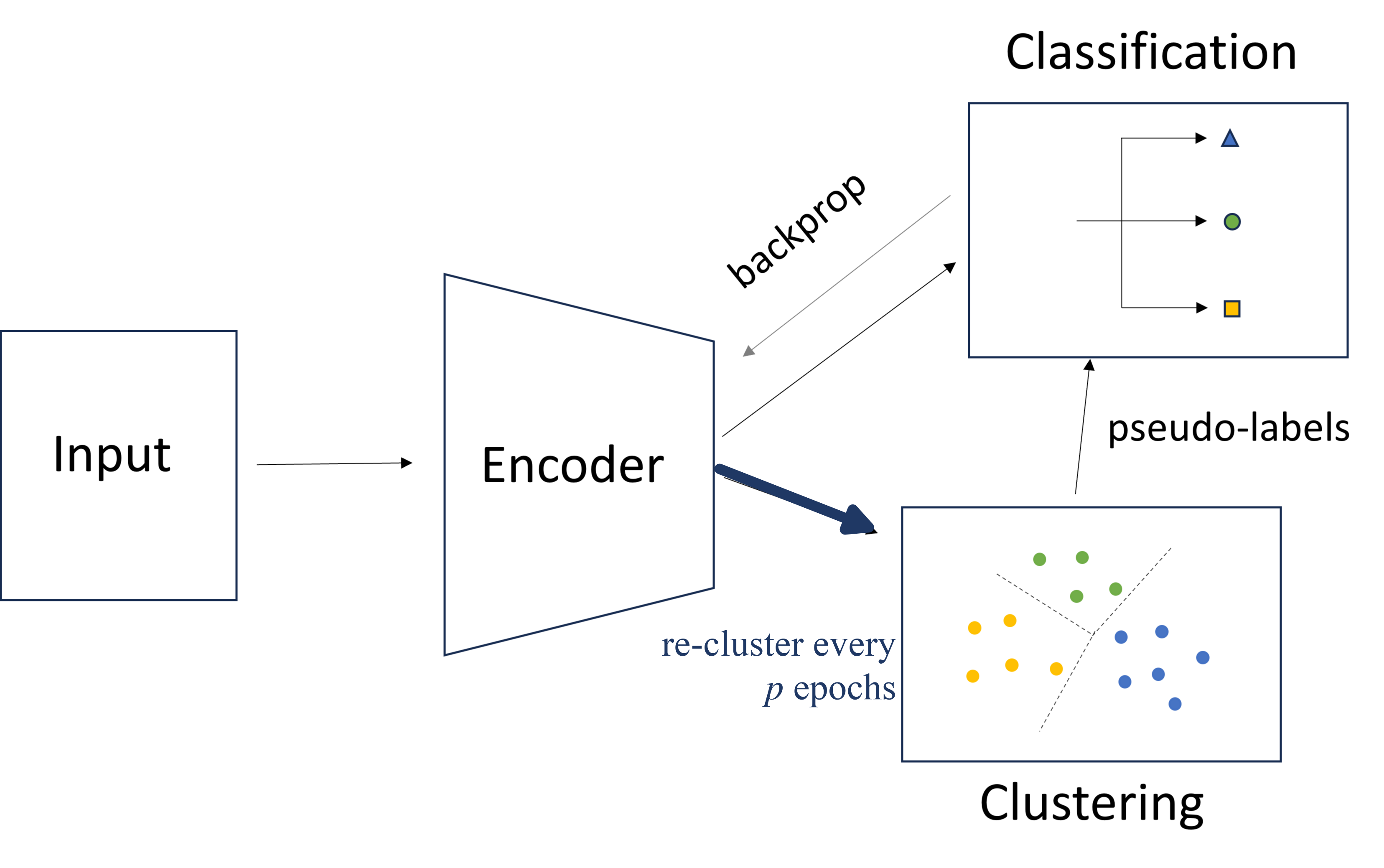}
\caption{Reassignment step: re-cluster the latent codes every $p$ epochs.}
\label{fig:reassign_step}
\end{subfigure}
\caption{The two alternating phases of training. The regular step refines the
encoder given fixed pseudo-labels; the reassignment step updates the
pseudo-labels given the improved encoder.}
\label{fig:steps}
\end{figure}

\begin{figure}[!htbp]
\centering
\includegraphics[width=0.50\linewidth]{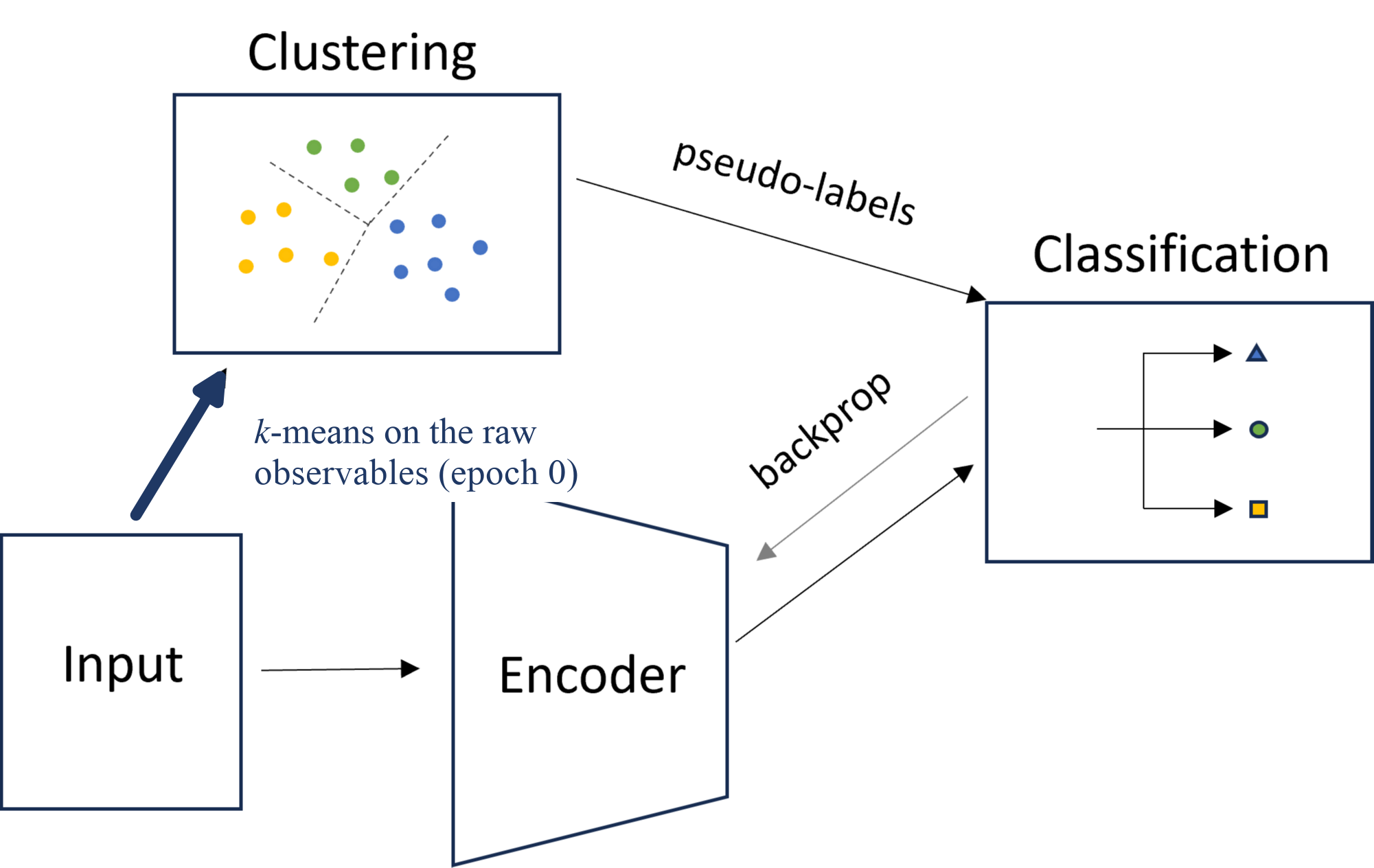}
\caption{Warmup: pseudo-labels are initialized by $k$-means on the
standardized observables, bypassing the randomly initialized encoder.}
\label{fig:warmup_phase}
\end{figure}

\subsection{Warmup}
\label{sec:warmup}

Without warmup the initial pseudo-labels come from $k$-means on the output of a
randomly initialized encoder, so the first epochs fit noise. The
\textit{warmup} variant instead initializes pseudo-labels by $k$-means on the
standardized observables (Fig.~\ref{fig:warmup_phase}), so the first
assignments already reflect physical plasma properties. This is the only place
where physical knowledge enters the pipeline; everything after it is
self-supplied.

\section{A Separability Bound for Neighborhood-Preserving Embeddings}
\label{sec:theory}

It is often assumed that two dimensions are insufficient to separate complex populations. Mathematically, this is false: a mapping that sends each cluster to a distinct point can separate any number of clusters in a 2D plane. The true obstruction is not the target dimension, but the objective function the embedding optimizes. This section formalizes this distinction.

\subsection{Setting and Notation}

Let $\mathcal{X} = \{x_1, \ldots, x_N\} \subset \mathbb{R}^{p}$ be the set of solar wind observations ($p = 7$), and let $\mathcal{C} = \{C_1, \ldots, C_k\}$ be a partition of $\mathcal{X}$ into $k$ non-empty clusters, where $\mathcal{C}(x)$ denotes the cluster containing $x$. An \emph{embedding} is a map $f : \mathcal{X} \to \mathbb{R}^{q}$, where $q$ is the target dimension ($q = 2$ for visualization baselines; $q = m$ for our latent encoder). The theoretical statements below hold for every $p$ and $q$, including $q \ge p$.

We use Euclidean distances in both spaces. Assuming pairwise distances $\|f(x) - f(x')\|$ are distinct for $x \neq x'$, nearest neighbors are unambiguous (the generic case for our data). Let $\kappa \in \mathbb{N}$ be a neighborhood size satisfying $\kappa < \min_i |C_i|$. 

Our argument relies on the \emph{cut fraction}, $\beta_\kappa$: the average proportion of an observation's $\kappa$ nearest neighbors that belong to a different cluster. It is zero when clusters are perfectly isolated and large when they interleave.

\begin{definition}[$\kappa$-NN graph]
For $Y = f(\mathcal{X})$, let $G_\kappa(f)$ be the directed graph on the index set $\{1, \ldots, N\}$ with an edge $n \to n'$ whenever $f(x_{n'})$ is among the $\kappa$ points of $Y \setminus \{f(x_n)\}$ closest to $f(x_n)$. It has exactly $\kappa N$ edges. Since vertices are indexed by observations, $\mathcal{C}$ labels the vertices of $G_\kappa(f)$ for any $f$.
\end{definition}

\begin{definition}[Cut fraction]
The \emph{cut fraction} of $\mathcal{C}$ in $G_\kappa(f)$ is
\[
\beta_\kappa(f, \mathcal{C}) \;=\; \frac{1}{\kappa N}
\big|\{\, n \to n' \in G_\kappa(f) \;:\;
\mathcal{C}(x_n) \neq \mathcal{C}(x_{n'}) \,\}\big|
\;\in\; [0, 1].
\]
This represents the fraction of neighbor links crossing cluster boundaries. Equivalently, $1 - \beta_\kappa$ is the $\kappa$-NN purity of the partition. We denote the input space cut fraction as $\beta_\kappa(\mathrm{id}, \mathcal{C})$, taking $f$ as the identity map.
\end{definition}

\begin{definition}[Neighborhood preservation]
An embedding $f$ is \emph{$\kappa$-neighborhood preserving} if $G_\kappa(f) = G_\kappa(\mathrm{id})$ as directed graphs. It is \emph{$\varepsilon$-approximately $\kappa$-neighborhood preserving} if their edge sets differ by at most $\varepsilon \kappa N$ edges.
\end{definition}

Neighborhood preservation is the theoretical goal of standard baseline embeddings: t-SNE matches neighbor probabilities, UMAP matches a fuzzy $k$-NN topology, and PHATE preserves diffusion neighborhoods. Because exact preservation is rarely achieved in practice, we state the approximate version.

\subsection{Two Bounds and an Escape}

\begin{proposition}[Exact preservation fixes the cut fraction]
\label{prop:one}
If $f$ is $\kappa$-neighborhood preserving, then
$\beta_\kappa(f, \mathcal{C}) = \beta_\kappa(\mathrm{id}, \mathcal{C})$ for every
partition $\mathcal{C}$. The same holds for any functional of the pair
$(G_\kappa, \mathcal{C})$, including the graph cut, conductance, modularity, and
spectral clustering objectives.
\end{proposition}

\begin{proof}
The edge set of $G_\kappa$ is unchanged by hypothesis, and the vertex labeling
$\mathcal{C}$ is carried by the index set rather than by position, so it is the
same for every $f$. Any functional taking these two as arguments,including
the cut fraction, therefore remains unchanged. \qed
\end{proof}

\begin{proposition}[Approximate preservation bounds the change]
\label{prop:approx}
If $f$ is $\varepsilon$-approximately $\kappa$-neighborhood preserving, then
$\big|\beta_\kappa(f, \mathcal{C}) - \beta_\kappa(\mathrm{id},
\mathcal{C})\big| \le \varepsilon$.
\end{proposition}

\begin{proof}
Let $E$ and $E'$ be the edge sets of $G_\kappa(\mathrm{id})$ and $G_\kappa(f)$,
and let $c(\cdot)$ count edges crossing cluster boundaries. Then
$c(E') - c(E) = c(E' \setminus E) - c(E \setminus E')$; the two terms are
non-negative and enter with opposite signs, and each is at most
$|E' \triangle E|$, so $|c(E') - c(E)| \le |E' \triangle E| \le
\varepsilon \kappa N$. Dividing by $\kappa N$ yields the bound. \qed
\end{proof}

Proposition~\ref{prop:approx} governs methods like t-SNE and UMAP. Their output
cut fraction cannot deviate from the input value by more than their neighborhood
distortion $\varepsilon$. They might \emph{lose} local structure, but they cannot
\emph{create} physical separation absent in the original metric. Margin
objectives, however, are not subject to this bound because they actively rewire
the neighbor graph.

\begin{proposition}[A margin removes cross-cluster neighbors]
\label{prop:two}
Suppose $f$ satisfies, for every $x \in \mathcal{X}$,
\begin{equation}
\max_{p \,\in\, \mathcal{C}(x)} \|f(x) - f(p)\| \;+\; \alpha
\;\le\;
\min_{n \,\notin\, \mathcal{C}(x)} \|f(x) - f(n)\|
\label{eq:margincond}
\end{equation}
for some $\alpha > 0$. Then $\beta_\kappa(f, \mathcal{C}) = 0$ for every
$\kappa < \min_i |C_i|$, and every observation lies at distance at least
$\alpha$ further from any other cluster than from any member of its own.
\end{proposition}

\begin{proof}
Fix $x$ and let $C = \mathcal{C}(x)$. Condition \eqref{eq:margincond} ensures
that all $|C| - 1$ points of $C \setminus \{x\}$ are close
point outside $C$. Since $\kappa < \min_i |C_i| \le |C|$, the $\kappa$ nearest
neighbors of $f(x)$ lie strictly within $C$, so no edge le
cluster. As $x$ was arbitrary, $\beta_\kappa(f, \mathcal{C}) = 0$. \qed
\end{proof}

Condition \eqref{eq:margincond} is the population-level equivalent of the
triplet margin in Eq.~\ref{eq:loss}: the loss vanishes on a triple $(x, p, n)$
exactly when $\|f(x) - f(p)\| + \alpha \le \|f(x) - f(n)\|$, and
\eqref{eq:margincond} requires this on \emph{every} triple, not only on those
drawn during training. Training Solar-CDC is therefore a stochastic surrogate
for this hypothesis, driving $\beta_\kappa$ toward zero as the loss vanishes and
attaining it only in the limit.

\begin{corollary}[The Separability Gap]
\label{cor:gap}
If initial cluster overlap exceeds the embedding's distortion
($\beta_\kappa(\mathrm{id}, \mathcal{C}) > \varepsilon$), no
$\varepsilon$-approximate neighborhood-preserving method can achieve perfect
separation ($\beta_\kappa = 0$). In contrast, a margin-satisfying embedding
always achieves $\beta_\kappa = 0$.
\end{corollary}

\begin{proof}
Immediate from Propositions~\ref{prop:approx} and~\ref{prop:two}: the former
bounds the cut fraction strictly above zero
($\beta_\kappa(f, \mathcal{C}) \ge \beta_\kappa(\mathrm{id}, \mathcal{C})
- \varepsilon > 0$), while the latter guarantees it reaches zero. \qed
\end{proof}

\subsection{The Hypothesis Holds on These Data}

The physical reality of the slow solar wind matches the premise of Corollary~\ref{cor:gap} (Sec.~\ref{sec:data}). Because the two slow-wind channels overlap in bulk speed and density, their observations frequently appear as mutual nearest neighbors in the raw observable space.

Table~\ref{tab:beta} validates this at $\kappa = 10$. In the raw input space the cut fraction is $\beta_{10} = 0.064$. The two embeddings that come closest to preserving the neighbor graph, t-SNE and UMAP, return $0.071$ and $0.084$, within the range Proposition~\ref{prop:approx} allows. PHATE and PCA drift further, to $0.148$ and $0.174$, because they preserve diffusion distance and variance rather than the graph itself and are covered by neither proposition. None of the four moves the cut fraction downward.

In contrast, methods with margin objectives (TriMap and Solar-CDC) are free to rewrite the neighbor graph, and they demonstrate that this freedom is necessary but not sufficient. TriMap heavily rewires the graph ($\beta_{10} = 0.407$) toward an unphysical partition of its own making (see Sec.~\ref{sec:baselines}). Solar-CDC, however, successfully drives the cut fraction down to $0.001$, a sixty-fold reduction completely forbidden to neighborhood-preserving maps by Proposition~\ref{prop:one}.

\begin{table}[!htbp]
\centering
\small
\caption{Cut fraction $\beta_{10}$ ($\kappa = 10$, $8{,}000$-point subsample).
As Props.~\ref{prop:one}--\ref{prop:approx} require, neighborhood-preserving
methods stay within $\varepsilon$ of the input metric, and none improves on it; Solar-CDC reduces it to near zero.}

\label{tab:beta}
\begin{tabular*}{\linewidth}{@{\extracolsep{\fill}}lcc}
\toprule
Space & $\beta_{10}$, recovered partition & $\beta_{10}$, threshold partition \\
\midrule
Input, seven observables & 0.064 & 0.131 \\
\midrule
\multicolumn{3}{l}{\textit{neighborhood preserving}} \\
t-SNE \cite{tsne} ($q = 2$)   & 0.071 & 0.158 \\
UMAP \cite{umap} ($q = 2$)    & 0.084 & 0.197 \\
PHATE \cite{phate} ($q = 2$)  & 0.148 & 0.228 \\
PCA \cite{pca_auto} ($q = 2$) & 0.174 & 0.215 \\
\midrule
\multicolumn{3}{l}{\textit{free to change the neighbor graph}} \\
TriMap \cite{trimap} ($q = 2$) & 0.407 & 0.439 \\
Solar-CDC latent  & \textbf{0.001} & 0.261 \\
\bottomrule
\end{tabular*}
\end{table}

\begin{remark}[What the result does not say]
Condition~\eqref{eq:margincond} is stated with respect to a given partition, and the objective drives $\beta_\kappa$ toward zero for \emph{any} partition it is handed, even an unphysical one. Table~\ref{tab:beta} (right column) illustrates this directly: when evaluated against the unseen threshold partition, the learned representations perform worse than the raw input metric. This explains why a margin objective \emph{can} separate populations where neighborhood-preserving maps fail, but crucially, it is the initialization of the pseudo-labels that dictates \emph{which} partition is ultimately separated. 

For the dual-origin problem, this is exactly why the warmup phase matters. Solar-CDC initializes pseudo-labels using $k$-means directly on the raw composition ratios, which act as frozen-in tracers of the coronal source. Consequently, the partition the margin sharpens is physically meaningful from the outset. Sec.~\ref{sec:external} tests whether that alignment survives.
\end{remark}

\section{Results}
\label{sec:results}

\textit{Setting.} We evaluate 30 configurations across a grid of attention heads ($\in \{1,2,4,8,16,32\}$) and latent dimensions ($m \in \{2,4,8,16,32\}$), training each for 200 epochs. This grid is run twice: once using solely the seven physical observables, and once with two appended metadata columns (Sec.~\ref{sec:cols}). All 60 runs use the physical warmup and a pseudo-label update period of $p = 4$. Training averages $167$~seconds per configuration on an Apple M4 Max using the PyTorch Metal backend.

\textit{Clustering Quality Evaluation.}
\label{sec:protocol}
Cluster quality is measured by the silhouette coefficient, computed on the latent codes using the partition predicted by the classification head. For a point $x$ with mean intra-cluster distance $a(x)$ and mean nearest-cluster distance $b(x)$, the coefficient is $\big(b(x)-a(x)\big)/\max\{a(x),b(x)\}$. Averaged across all points, this yields a score in $[-1,1]$, where higher values indicate tighter, better-separated clusters.

The
hyperparameter grid is trained on an $80\%$ seeded split ($24{,}482$
observations) and scored exactly on the held-out $20\%$ ($6{,}120$). All subsequent experiments are fitted and scored exactly on the full $30{,}602$ observations, with one exception. Because t-SNE, UMAP, PHATE and TriMap scale superlinearly in $N$, the two tables that compare against them use a fixed $8{,}000$-point subsample. In the cut-fraction comparison of
Table~\ref{tab:beta} that subsample is the same for every method, ours
included. In the silhouette comparison of Table~\ref{tab:baselines} it covers
the baselines only: the Solar-CDC entry there is the grid figure, scored on the
held-out split, which is one further reason that comparison is indicative rather
than exact. 

We present two types of comparisons. Internal ablations of Solar-CDC follow identical protocols and are directly comparable. Conversely, dimensionality-reduction baselines are scored inside their own generated representations: the standard reporting protocol for these methods. Because the evaluation spaces differ, comparisons between Solar-CDC and these baselines are indicative rather than exact.

\textit{Solar-CDC Hyperparameter Grid.} Table~\ref{tab:grid} details the silhouette scores across the $6\times5$ grid of attention heads and latent dimensions. All 30 runs successfully converged to three non-empty clusters. The optimal configuration uses a single attention head and a 2-D latent code, achieving a silhouette of $0.869$ (grid mean $0.795$). As expected for a distance-based metric, the silhouette score decreases monotonically as the latent dimension increases.

To rigorously evaluate these parameters, we analyze the grid as a randomized block design using the non-parametric Friedman test. The latent dimension strongly dictates performance ($\chi^2 = 22.3$, $p = 1.8\times10^{-4}$, using head counts as blocks). In contrast, the number of attention heads has no significant effect ($\chi^2 = 7.4$, $p = 0.19$, using latent dimensions as blocks). We conclude that a single attention head suffices; the encoder's architectural capacity is not the binding constraint for this task.

\begin{table}[tbp]
\centering
\small
\caption{Silhouette over the Solar-CDC hyperparameter grid (30{,}602 observations, seven
observables, $k=3$, 200 epochs, pseudo-label period $p=4$, Euclidean triplet
distance). All 30 runs converged to three non-empty clusters.}
\label{tab:grid}
\begin{tabular*}{\linewidth}{@{\extracolsep{\fill}}lccccc}
\toprule
& \multicolumn{5}{c}{Latent dimension $m$} \\
\cmidrule(lr){2-6}
Heads & 2 & 4 & 8 & 16 & 32 \\
\midrule
 1 & \textbf{0.869} & 0.828 & 0.820 & 0.755 & 0.725 \\
 2 & 0.829 & 0.806 & 0.773 & 0.748 & 0.764 \\
 4 & 0.866 & 0.814 & 0.792 & 0.752 & 0.728 \\
 8 & 0.840 & 0.821 & 0.809 & 0.756 & 0.772 \\
16 & 0.832 & 0.834 & 0.779 & 0.775 & 0.740 \\
32 & 0.853 & 0.823 & 0.797 & 0.780 & 0.760 \\
\bottomrule
\end{tabular*}
\end{table}

\subsection{Only the Physical Observables Are Needed.} 
\label{sec:cols}
The Solar Orbiter archive \cite{zhao2024classification} contains two metadata columns alongside the seven physical observables: observation time and a previously assigned class label. To rigorously test their utility, we repeated the entire 30-configuration grid with both columns appended to the input.

Adding this metadata consistently degrades performance. The silhouette score dropped in 23 of the 30 configurations (an average decrease of $0.017$), and the optimal configuration's score fell from $0.869$ to $0.857$. Because each configuration serves as its own control, we can confirm this penalty is highly statistically significant across the matched pairs (Wilcoxon signed-rank $p = 1.2\times10^{-3}$; paired $t$-test $p = 6.9\times10^{-4}$). 

We conclude that the metadata columns introduce a small but consistent penalty rather than a benefit. Consequently, all subsequent experiments rely strictly on the seven physical observables. Crucially, no learned result reported in this paper depends on the archived class label.

\begin{table}[!htbp]
\centering
\small
\caption{Silhouette of each dimensionality-reduction baseline inside its own
two-dimensional embedding (8{,}000-point subsample; every method here is
superlinear in $N$), against Solar-CDC in its learned representation. The
grouping is the one Sec.~\ref{sec:theory} predicts: the methods that aim to
preserve neighborhoods sit together, and the one that optimizes triplet order
does not.}
\label{tab:baselines}
\begin{tabular*}{\linewidth}{@{\extracolsep{\fill}}lccc}
\toprule
Reduction & K-Means & GMM & Agglomerative \\
\midrule
\multicolumn{4}{l}{\textit{neighborhood preserving}} \\
PCA \cite{pca_auto} (2-D)      & 0.400 & 0.360 & 0.414 \\
t-SNE \cite{tsne} ($p{=}30$)   & 0.428 & 0.427 & 0.388 \\
t-SNE ($p{=}100$)              & 0.438 & 0.424 & 0.422 \\
t-SNE ($p{=}200$)              & 0.450 & 0.440 & 0.372 \\
UMAP \cite{umap} (15 nn)       & 0.443 & 0.452 & 0.420 \\
UMAP (30 nn)                   & 0.445 & \textbf{0.454} & 0.429 \\
UMAP (50 nn)                   & 0.453 & 0.443 & 0.434 \\
PHATE \cite{phate} (15 nn)     & 0.448 & 0.419 & 0.383 \\
PHATE (30 nn)                  & 0.450 & 0.414 & 0.434 \\
PHATE (50 nn)                  & 0.451 & 0.413 & 0.420 \\
\midrule
\multicolumn{4}{l}{\textit{triplet based}} \\
TriMap \cite{trimap} (8 inliers) & 0.795 & 0.662 & \textbf{0.824} \\
TriMap (12 inliers)              & 0.753 & 0.659 & 0.821 \\
TriMap (20 inliers)              & 0.752 & 0.628 & 0.668 \\
\midrule
\multicolumn{4}{l}{Solar-CDC, learned representation \hfill \textbf{0.869}} \\
\bottomrule
\end{tabular*}
\end{table}

\subsection{Comparing Solar-CDC with Dimensionality-Reduction Clustering Baselines}
\label{sec:baselines}

Table~\ref{tab:baselines} evaluates the standard pipeline: reduce the seven observables to two dimensions, then cluster. To reflect standard reporting practice, each method is scored by its silhouette coefficient computed strictly inside its own 2-D embedding.

The results split exactly along the theoretical fault line drawn in Sec.~\ref{sec:theory}. The thirty neighborhood-preserving combinations (ten reductions each followed by three clusterings) stall in a narrow band ($0.426 \pm 0.025$, peaking at $0.454$). This is what Props.~\ref{prop:one}--\ref{prop:approx} lead one to expect. An embedding built to keep the neighbor graph cannot change the cut structure of a partition, so it has no mechanism for sharpening a boundary the input metric does not already carry. In stark contrast, Solar-CDC reaches $0.869$, and even its weakest configuration ($0.725$) strictly dominates the strongest baseline in this group. The two sets do not overlap at all (Cliff's $\delta = 1.00$, Mann-Whitney $p = 1.5\times10^{-11}$).

TriMap provides an informative exception. Because it optimizes triplet order rather than neighborhood structure \cite{trimap}, it is exempt from the bound of Prop.~\ref{prop:one} and achieves a highly competitive silhouette of $0.824$. Solar-CDC stays ahead of its nine configurations by a smaller margin (Mann-Whitney $p = 0.012$, Cliff's $\delta = 0.50$: a random Solar-CDC run outscores a random TriMap run three times out of four, rather than always). TriMap's high score is exactly what the theory predicts: the deciding factor is not the target dimension (which is two for every baseline), but whether the objective is free to rewrite the neighbor graph.

\textit{Solar-CDC is Physically Aware.} Escaping the theoretical bound, as TriMap does, does not guarantee recovering the physics. To measure physical accuracy here and in Sec.~\ref{sec:external}, we construct a three-way reference label of our own. It applies established $\mathrm{O}^{7+}/\mathrm{O}^{6+}$ thresholds from the heliophysics literature \cite{xu2015new,damicis2021alfvenic,lepri2013solar} directly to the raw observations, and it ignores the class label shipped with the archive. Evaluated against this physical baseline, TriMap's partition fails outright. It yields an NMI of $0.036$, a matched accuracy of $0.465$ (worse than a constant labeling achieves), and an adjusted Rand index of $-0.021$ (slightly worse than random chance), all despite its impressive $0.824$ silhouette. Ironically, the neighborhood-preserving embeddings, with silhouettes half as large, capture the physics much better (NMI $0.33$--$0.42$). TriMap proves that a geometrically tight partition is not necessarily a physically meaningful one. Solar-CDC is the only method in the comparison that successfully achieves both.

\subsection{Number of Solar Wind Clusters}
\label{sec:k}

Table~\ref{tab:k} evaluates $k \in \{2,3,4,5\}$ using ten random seeds per value. The cluster count drives significant variation (Kruskal-Wallis $p = 3.2\times10^{-5}$). Both $k = 4$ and $k = 5$ perform strictly worse than $k = 3$ (Welch $p < 4\times10^{-4}$), so the data do not support dividing the wind more finely than three ways.

Crucially, however, the silhouette coefficient cannot statistically distinguish $k = 2$ from $k = 3$ ($\Delta\mu = 0.012$ vs.\ seed spread $0.033$, $p = 0.42$). This statistical dead heat is the dual-origin problem in a nutshell. A purely geometric index is dominated by the primary fast/slow split, which exhibits massive differences in speed and density. The secondary split, separating the two slow-wind streams, is purely compositional and kinematically invisible, contributing almost nothing to the distance metric. Had a third cluster been geometrically obvious, the origin of the slow wind would not have remained an open question for twenty years.

We therefore adopt $k = 3$ strictly on the physical grounds of Sec.~\ref{sec:intro-physics}: the existence of three coronal source regions with distinct freeze-in temperatures. While geometry alone cannot break the tie between $k=2$ and $k=3$, the three-way partition is supported after the fact by its agreement with the external composition taxonomy (Sec.~\ref{sec:external}).

\begin{table}[!htbp]
\centering
\small
\caption{Cluster-count sweep, ten seeds per $k$ (mean $\pm$ s.d.), with the
silhouette computed exactly on all $30{,}602$ observations. Both $k = 2$ and
$k = 3$ are significantly better than $k = 4$ and $k = 5$; they are not
distinguishable from each other.}
\label{tab:k}
\begin{tabular*}{\linewidth}{@{\extracolsep{\fill}}ccl}
\toprule
$k$ & Silhouette & vs.\ $k = 3$ \\
\midrule
2 & $0.856 \pm 0.033$ & $p = 0.42$ \\
3 & $0.843 \pm 0.035$ & --- \\
4 & $0.765 \pm 0.044$ & $p = 3.9\times10^{-4}$ \\
5 & $0.688 \pm 0.089$ & $p = 2.6\times10^{-4}$ \\
\bottomrule
\end{tabular*}
\end{table}

\subsection{Agreement with the Published Composition Taxonomy}
\label{sec:external}

A high silhouette score establishes geometric tightness, but as TriMap demonstrates (Table~\ref{tab:baselines}), tight clusters need not be physical. To test whether Solar-CDC recovers solar wind physics rather than convenient geometry, we further evaluate it against an independent, domain-specific reference.

The heliophysics literature classifies solar wind using strictly defined $\mathrm{O}^{7+}/\mathrm{O}^{6+}$ thresholds: $<0.145$ for the coronal-hole interior \cite{xu2015new}, $0.145$--$0.20$ for the coronal-hole boundary \cite{damicis2021alfvenic}, and $>0.20$ for the streamer belt \cite{lepri2013solar}. Applying these thresholds to our dataset yields reference labels in three classes of $15{,}321$, $5{,}309$, and $9{,}972$ observations, respectively. We score the learned partitions against them using Normalized Mutual Information (NMI), the Adjusted Rand Index (ARI), and Matched Accuracy. For context, random assignment yields an NMI of $0.000$ and $0.337$ accuracy, while predicting the majority class yields $0.501$ accuracy.

We must address an obvious confounder first: because $\mathrm{O}^{7+}/\mathrm{O}^{6+}$ is one of the seven model inputs, evaluating against a reference defined by it is inherently circular. To separate learned physics from the mere isolation of a single input feature, we introduce a \emph{6-observable control} in the lower half of Table~\ref{tab:external}. Here, the model is retrained with $\mathrm{O}^{7+}/\mathrm{O}^{6+}$ withheld from the input entirely, then scored against the same reference labels.

\begin{table}[!htbp]
\centering
\small
\caption{Agreement with the external $\mathrm{O}^{7+}/\mathrm{O}^{6+}$ taxonomy of \cite{xu2015new,damicis2021alfvenic,lepri2013solar}, defined by published thresholds rather than by archived class labels, and unseen during training. Solar-CDC values are means over five seeds (16 heads, $m = 2$). The 6-observable control is trained without $\mathrm{O}^{7+}/\mathrm{O}^{6+}$ and scored against the same reference.}
\label{tab:external}
\begin{tabular*}{\linewidth}{@{\extracolsep{\fill}}llccc}
\toprule
Method & Inputs & NMI & ARI & Matched acc. \\
\midrule
Random labeling             & ---   & 0.000 & 0.000 & 0.337 \\
Largest class only          & ---   & 0.000 & 0.000 & 0.501 \\
TriMap \cite{trimap} + Agglom.\  & seven & 0.036 & $-0.021$ & 0.465 \\
\midrule
$k$-means                   & seven & 0.348 & 0.284 & 0.628 \\
Solar-CDC                   & seven & $\mathbf{0.422 \pm 0.079}$ &
                                      $0.385 \pm 0.088$ &
                                      $\mathbf{0.706 \pm 0.046}$ \\
\midrule
$k$-means, control          & six   & 0.311 & 0.259 & 0.613 \\
Solar-CDC, control          & six   & $0.321 \pm 0.019$ & $0.307 \pm 0.019$ &
                                      $0.653 \pm 0.015$ \\
\bottomrule
\end{tabular*}
\end{table}

\textit{Two Results from the External Check.} Two key results emerge from Table~\ref{tab:external}. First, on the full 7-observable dataset, Solar-CDC outperforms standard $k$-means by $0.074$ in NMI and $0.078$ in matched accuracy. The accuracy gap strictly holds up against the seed-to-seed spread ($p = 0.019$ over five seeds). The NMI gap is less strictly significant ($p = 0.105$) because NMI varies more heavily between seeds ($\pm 0.079$).

Second, and more importantly, the 6-observable control succeeds. Blinded to the very feature the reference is built on, Solar-CDC still reaches an NMI of $0.321$ and an accuracy of $0.653$, far above the naive floors, and does so with a quarter of the original seed-to-seed variance. The remaining six observables therefore carry enough entangled compositional structure to recover the true charge-state boundaries. This physical robustness is not unique to our architecture: standard $k$-means on the same six observables reaches an NMI of $0.311$, statistically indistinguishable from Solar-CDC ($p = 0.30$), though Solar-CDC retains a clear edge in matched accuracy ($p = 0.004$). 

Ultimately, the control establishes that the recovered partition is driven by the joint composition space rather than the memorization of a single ratio. In that sense, the clusters Solar-CDC finds are not merely geometrically tight; they are physically real.

\begin{figure}[!htbp]
\centering
\includegraphics[width=0.88\linewidth]{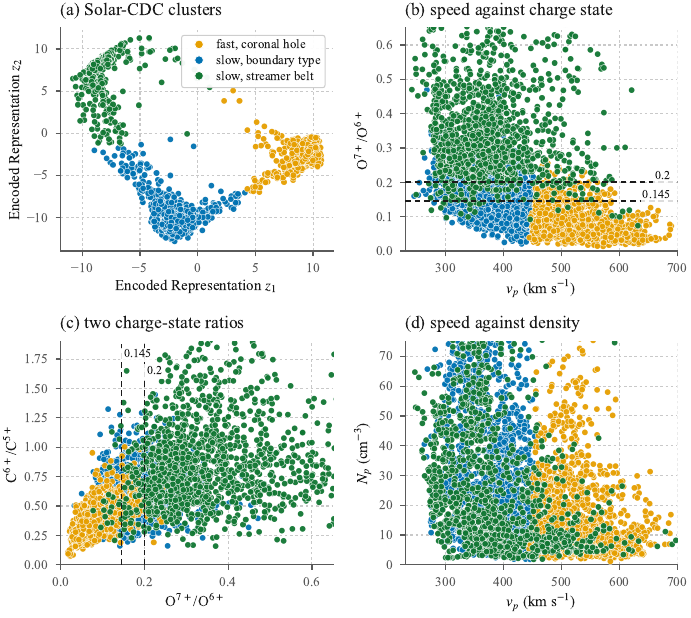}
\caption{The three recovered clusters in learned and physical spaces (colors consistent throughout; $6{,}000$ random observations sampled for legibility). \textbf{(a)} The 2-D latent codes. \textbf{(b)} Bulk speed vs.\ $\mathrm{O}^{7+}/\mathrm{O}^{6+}$, with dashed lines marking the published coronal-hole ($0.145$) and streamer-belt ($0.20$) thresholds. The two slow-wind populations share an identical speed range but divide sharply at the compositional boundary. \textbf{(c)} Correlated charge-state ratios. The separation spans the broader composition space rather than relying on a single variable. \textbf{(d)} Bulk speed vs.\ density. The two slow clusters perfectly interleave. Together, panels (b) and (d) visualize the core premise of our method: the distinct origins of the slow wind are clearly separable by composition, yet completely invisible to classical kinematic properties.}
\label{fig:latent}
\end{figure}

\begin{table}[!htbp]
\centering
\small
\caption{Mean plasma properties per cluster, ordered by charge state (16 heads, $m = 2$). Cluster~0 isolates the fast wind, while Clusters~1 and~2 partition the slow wind. This compositional profile represents the 11 of 24 runs that successfully resolve the boundary population. For physical context, established $\mathrm{O}^{7+}/\mathrm{O}^{6+}$ thresholds are: $<0.145$ (coronal-hole interior), $0.145$--$0.20$ (boundary), and $>0.20$ (streamer belt) \cite{lepri2013solar,xu2015new,damicis2021alfvenic}. \textit{Note:} Silhouette scores here are computed across the full dataset rather than the held-out split, precluding direct comparison with Table~\ref{tab:grid}.}

\label{tab:clusters}
\begin{tabular*}{\linewidth}{@{\extracolsep{\fill}}lrrrrrrr}
\toprule
Cluster & $n$ & $\bar v_p$ & $\bar N_p$ &
$\overline{\mathrm{O}^{7+}\!/\mathrm{O}^{6+}}$ &
$\overline{\mathrm{C}^{6+}\!/\mathrm{C}^{4+}}$ &
$\overline{\mathrm{Fe/O}}$ & $\overline{Q_O}$ \\
\midrule
0\ coronal hole  & 10{,}415 & 526 & 14.2 & 0.080 & 2.14 & 0.146 & 6.07 \\
1\ boundary type & 12{,}581 & 374 & 28.1 & 0.160 & 4.37 & 0.154 & 6.15 \\
2\ streamer belt & \phantom{0}7{,}606 & 379 & 32.3 & 0.400 & 11.00 & 0.220 & 6.39 \\
\bottomrule
\end{tabular*}
\end{table}

\subsection{Physical Interpretation of the Clusters}

Table~\ref{tab:clusters} and Figs.~\ref{fig:latent} and~\ref{fig:clusters} detail the mean plasma properties of the three recovered clusters. When ordered by charge state, the progression is strictly monotonic across every compositional variable, mirroring the ordering of the three known coronal source regions.

The first population represents the fast, tenuous, and cool-sourced wind: its mean $\mathrm{O}^{7+}/\mathrm{O}^{6+}$ ratio is $0.080$, safely below the $0.145$ coronal-hole threshold. The second population is slow and dense, averaging $0.160$, squarely inside the $0.145$--$0.20$ window predicted for coronal-hole-boundary plasma \cite{damicis2021alfvenic}. The third population is similarly slow but originates from a much hotter source, registering at $0.400$, double the $0.20$ streamer-belt threshold.

The contrast between Clusters~1 and~2 directly exposes the dual-origin problem. They differ by a negligible $4$~km/s in bulk speed, yet by a massive factor of $2.5$ in their $\mathrm{O}^{7+}/\mathrm{O}^{6+}$ ratios. This is precisely the regime where classical kinematic thresholds collapse, and it shows that the two slow-wind streams can be disentangled only through their frozen-in composition.

\begin{figure}[!htbp]
\centering
\includegraphics[width=0.90\linewidth]{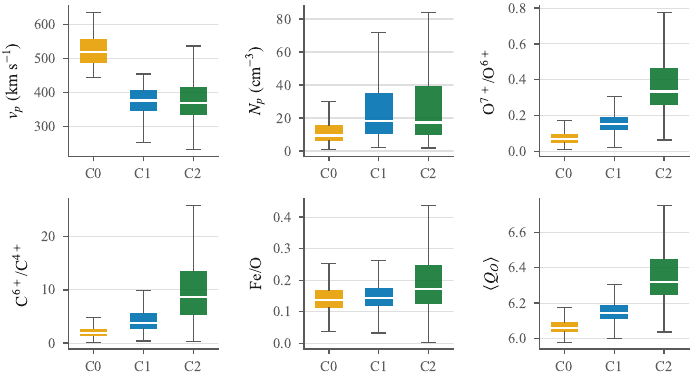}
\caption{Distribution of each observable per cluster (boxes: quartiles;
whiskers: $1.5\times$IQR; outliers omitted). Separation is monotone and largest
in the charge-state variables, and smallest in bulk speed, where clusters 1 and
2 overlap almost completely.}
\label{fig:clusters}
\end{figure}

\section{Discussion}

The preceding subsections establish a partition that is geometrically tight, aligns with an independent physical taxonomy, and survives the ablation of the exact variable defining that taxonomy. This section outlines the physical implications and limitations of these results.

\textit{Evidence for the Dual-Origin Picture.} The recovered partition supports the dual-origin hypothesis: an intermediate-composition slow-wind population exists, and it is cleanly separable from the streamer-belt wind by composition alone, despite overlapping almost entirely in bulk speed. Crucially, this compositional evidence is extracted completely without labels. It does not, however, establish that the population is \emph{Alfv\'enic}. Alfv\'enicity is determined by magnetic and velocity fluctuations, which are absent from our seven observables. We therefore conservatively designate Cluster~1 as \emph{boundary type}, and treat its identity as Alfv\'enic slow wind as a plausible identification that the present data cannot test.

\textit{Interpreting the Silhouette Scores.} Following standard reporting practice, every silhouette score in this paper is computed within the specific representation space each model learns. Consequently, Table~\ref{tab:baselines} measures how effectively each method separates populations within its \emph{own} generated embedding, rather than comparing them on a shared metric footing. We adopt this protocol to ensure fair, published-standard comparisons against the dimensionality-reduction baselines, but explicitly note that cross-method silhouette comparisons remain indicative rather than absolute.

\textit{Limitations.}
(1)~\textit{Missing magnetic data:} as noted, incorporating magnetic fluctuation statistics to measure Alfv\'enicity directly is the single most valuable extension.
(2)~\textit{Model selection:} configuration choice currently relies on the external reference labels. Developing a fully unsupervised criterion that reliably correlates with physical correctness remains an open challenge.
(3)~\textit{Reference baseline:} the external taxonomy relies on published $\mathrm{O}^{7+}/\mathrm{O}^{6+}$ thresholds \cite{xu2015new,damicis2021alfvenic,lepri2013solar}, so agreement measures consistency with the current literature, not absolute ground truth.
(4)~\textit{Temporal scope:} the dataset spans Solar Orbiter's observations during the rising phase of solar cycle 25. Generalization across solar cycles and missions remains untested.
(5)~\textit{Architectural excess:} the best configuration requires only a single attention head and a low-dimensional latent space, meaning encoder capacity is not the active ingredient here; a simpler architecture may suffice for this task.
 
% =============================================================================
\section{Conclusion}

We presented Solar-CDC, a self-supervised contrastive deep clustering method for in-situ solar wind composition, evaluating it on $30{,}602$ Solar Orbiter observations entirely without labels at fitting time. It achieves a silhouette score of $0.869$ in its learned representation, vastly outperforming thirty combinations of dimensionality reduction and clustering (which peak at $0.454$ with zero distribution overlap). Furthermore, the recovered clusters are ordered monotonically across every composition variable, perfectly mirroring the established charge-sequence of coronal source regions.

These insights extend far beyond this specific archive. The performance gap between learned metrics and standard 2-D embeddings is not constrained by target dimensionality, but by the objective function. An embedding designed to preserve neighborhood graphs cannot reduce a partition's cut fraction below its input-metric baseline, regardless of available dimensions. In contrast, a margin objective can rewrite the graph and drive the cut fraction to zero. This principle applies generally whenever populations of interest are mutual neighbors in measured space. That is the standard scenario when discriminating signals are distributed across multiple observables. Crucially, because a margin objective sharpens whichever partition it receives, initialization dictates success: our warmup phase, anchored directly in physical plasma measurements rather than random initialization, ensures the correct partition is targeted.

\textit{Future Work.} Three natural extensions emerge. \textit{(1)~Broader in-situ archives:} Historic missions such as ACE (1998--2011), Wind (1995--2004), and Ulysses (1992--2007) provide complementary composition suites, offering a rigorous test of how structural recovery degrades as observables are removed across three solar cycles and diverse heliospheric vantage points. \textit{(2)~Multimodal integration:} Incorporating magnetic field vectors and fluctuation statistics will directly test the hypothesis linking our boundary-type population to Alfv\'enic slow wind, converting a plausible physical identification into a direct measurement. \textit{(3)~Cross-domain applications:} Because our theoretical bound (Sec.~\ref{sec:theory}) is domain-agnostic, it applies universally wherever populations are near-neighbors and the discriminating signal spans a joint space rather than a single coordinate. Spectral archives, geochemical assays, and single-cell measurements share this exact topological structure, and standard visualization tools should be expected to fail there for the exact same geometric reasons.

\vspace{0.5em}
{\small Data and code: \url{https://github.com/hank08819/Solar-CDC}}

\subsubsection{Acknowledgements}
This work is partially supported by NASA Grant 80NSSC22K1015, NSF 2229138, and
the McCollum Endowed Chair startup fund.

\bibliographystyle{splncs04}
\bibliography{main}

% =============================================================================
\appendix

\section{Supplementary Results}
\label{app:results}

Table~\ref{tab:gridfull} lists every configuration of the grid with the
seven observables and, for comparison, the same grid with the two additional
columns of Sec.~\ref{sec:cols} appended to the input.

\begin{table}[h]
\centering
\small
\caption{Silhouette over the full grid, with the seven observables and with
the two extra columns of Sec.~\ref{sec:cols} appended to the input
(30{,}602 observations, $k=3$, 200 epochs, $p=4$).}
\label{tab:gridfull}
\begin{tabular*}{\linewidth}{@{\extracolsep{\fill}}lccccc@{\hspace{1.2em}}ccccc}
\toprule
& \multicolumn{5}{c}{seven observables} &
  \multicolumn{5}{c}{with time and label columns} \\
\cmidrule(lr){2-6}\cmidrule(lr){7-11}
Heads & 2 & 4 & 8 & 16 & 32 & 2 & 4 & 8 & 16 & 32 \\
\midrule
 1 & \textbf{0.869} & 0.828 & 0.820 & 0.755 & 0.725 & 0.857 & 0.803 & 0.811 & 0.746 & 0.687 \\
 2 & 0.829 & 0.806 & 0.773 & 0.748 & 0.764 & 0.842 & 0.783 & 0.804 & 0.753 & 0.730 \\
 4 & 0.866 & 0.814 & 0.792 & 0.752 & 0.728 & 0.853 & 0.786 & 0.790 & 0.721 & 0.714 \\
 8 & 0.840 & 0.821 & 0.809 & 0.756 & 0.772 & 0.850 & 0.762 & 0.796 & 0.750 & 0.689 \\
16 & 0.832 & 0.834 & 0.779 & 0.775 & 0.740 & 0.846 & 0.818 & 0.799 & 0.746 & 0.703 \\
32 & 0.853 & 0.823 & 0.797 & 0.780 & 0.760 & 0.839 & 0.797 & 0.800 & 0.747 & 0.714 \\
\bottomrule
\end{tabular*}
\end{table}

\end{document}